\documentclass{article}
\usepackage{fullpage}

\usepackage[utf8]{inputenc}
\usepackage{epsfig}
\usepackage{epstopdf}
\usepackage{graphicx}
\DeclareGraphicsExtensions{.pdf,.png,.jpg,.eps}

\usepackage{latexsym}
\usepackage{amsfonts}
\usepackage{bm}
\usepackage{algorithm} 
\usepackage{algpseudocode}
\usepackage{multirow}
\usepackage{url}
\usepackage{pifont}
\usepackage{hyperref}
\usepackage{cite}

\usepackage{amsmath,amssymb,amstext,amsthm}
\usepackage{enumitem}
\usepackage{verbatim}
\usepackage[usenames,dvipsnames]{color}

\newcommand{\MS}[1]{\ensuremath{\Delta}}

\newtheorem{theorem}{Theorem}
\newtheorem{lemma}{Lemma}[section]

\newtheorem{corollary}{Corollary}[section]

\newcommand{\bit}{\begin{myitemize}}
\newcommand{\eit}{\end{myitemize}}
\newcommand{\ben}{\begin{enum}}
\newcommand{\een}{\end{enum}}

\newcommand{\beq}[1]{\begin{equation}\label{#1}}
\newcommand{\eeq}{\end{equation}}

\newcommand{\ceil}[1]{\left \lceil #1 \right \rceil}
\newcommand{\floor}[1]{\left \lfloor #1 \right \rfloor}

\newcommand{\eps}{\varepsilon}

\renewcommand{\cal}{\mathcal}

\theoremstyle{remark}
\newtheorem*{remark}{Remark}

\newcommand{\adjacent}{\textsf{adjacent}}
\newcommand{\degree}{\textsf{degree}}
\newcommand{\neighbor}{\textsf{neighborhood}}
\newcommand{\spath}{\textsf{spath}}

\newcommand{\rank}{\textsf{rank}}
\newcommand{\select}{\textsf{select}}
\newcommand{\rmin}{\textsf{RMin}}
\newcommand{\rmax}{\textsf{RMax}}
\newcommand{\SUCC}{\textsf{SUCC}}

\title{Succinct Encodings for Families of Interval Graphs}
\begin{document}
	%\begin{frontmatter}
	
	\title{Succinct Data Structures for Families of Interval Graphs}
	\author{
        H\"{u}seyin~Acan\thanks{Research supported by a National Science Foundation Fellowship (Award No. 1502650).} \\
        Drexel University\\ \href{mailto:huseyin.acan@drexel.edu}{huseyin.acan@drexel.edu}
        %\textsuperscript{*}\thanks{* Research supported by a National Science Foundation Fellowship (Award No. 1502650).}
        \\\\
        Sankardeep~Chakraborty\thanks{This work is supported by JSPS KAKENHI Grants Number 18H05291.} \\ 
        National Institute of Informatics \\ \href{mailto:sankardeep.chakraborty@gmail.com}{sankardeep.chakraborty@gmail.com}
        %\cfootnote{Research partially supported by JST CREST Grant Number JPMJCR1402.} 
        \\\\
        Seungbum~Jo\thanks{Work done while the author was at Universit\'e libre de Bruxelles. Research supported by Fonds de la Recherche Scientifique-FNRS under Grant no MISU F 6001 1.} \\
        Chungbuk National University\\ \href{mailto:sbjo@chungbuk.ac.kr}{sbjo@chungbuk.ac.kr}
        \\\\
        Srinivasa~Rao~Satti\\
        Seoul National University\\
    	   \href{mailto:ssrao@cse.snu.ac.kr}{ssrao@cse.snu.ac.kr}.
}
	\date{}
	\maketitle
\begin{abstract}
We consider the problem of designing succinct data structures for {\it interval graphs} with $n$ vertices while supporting degree, adjacency, neighborhood and shortest path queries in optimal time.
Towards showing succinctness, we first show that at least $n\log_2{n} - 2n\log_2\log_2 n - O(n)$ bits
are necessary to represent any unlabeled interval graph $G$ with $n$ vertices, answering an open problem of Yang and Pippenger [Proc. Amer. Math. Soc. 2017]. This is augmented by a data structure of size $n\log_2{n} +O(n)$ bits while supporting not only the above queries optimally but also capable of executing various combinatorial algorithms (like proper coloring, maximum independent set etc.) on interval graphs efficiently. Finally, we extend our ideas to other variants of interval graphs, for example, {\it proper/unit interval graphs, $k$-improper interval graphs, and circular-arc graphs}, and design succinct data structures for these graph classes as well along with supporting queries on them efficiently.
% \PACS{PACS code1 \and PACS code2 \and more}
% \subclass{MSC code1 \and MSC code2 \and more}
\end{abstract}
\section{Introduction}
A simple undirected graph $G$ is called an {\it interval graph} if its vertices can be assigned to intervals on the real line so that two vertices are adjacent in $G$ if and only if their assigned intervals intersect. The set of intervals assigned to the vertices of G
is called a {\it realization} of $G$. 
%If the set of intervals can be chosen to be inclusion-free, then $G$ is called a {\it proper interval graph}. 
These graphs were first introduced by Haj\'{o}s~\cite{Hajos} who also asked for the characterization of them. The same problem was also asked, independently, by~\cite{Benser} while studying the structure of genes. Interval graphs naturally appear in a variety of contexts, for example, operations research and scheduling theory~\cite{Bar-NoyBFNS01}, biology especially in physical mapping of DNA~\cite{bio}, temporal reasoning~\cite{GolumbicS93} and many more. We refer the reader to~\cite{Golumbic,Golumbic85} for a thorough treatment of interval graphs and its applications. Eventually answering the question of Haj\'{o}s~\cite{Hajos}, several researchers came up with different characterizations of interval graphs, including linear time algorithms for recognizing them; see, for example,~\cite[Chapter 8]{Golumbic} for characterizations, and~\cite{BL} and~\cite{HMPV} for linear time algorithms. Moreover, by exploiting the special structure of interval graphs, many otherwise NP-hard problems in general graphs are also shown to have polynomial time algorithms for interval graphs~\cite{Golumbic85}. These include computing maximum independent set, reporting a proper coloring, returning a maximum clique etc. In spite of having many applications in practically motivated problems, we are not aware of any study of interval graphs from the point of view of \textit{succinct data structures}. The goal here is to store a set $Z$ of objects using the information theoretic minimum $\log(|Z|)+o(\log(|Z|))$ bits of space\footnote{throughout the paper, we use $\log$ to denote the logarithm to the base 2} along with supporting relevant set of queries efficiently, which we focus on in this paper. We assume the usual model of computation, namely a $\Theta (\log n)$-bit word RAM model where $n$ is the size of the input. 

\subsection{Related Work}
%\textbf{I think more results for interval graphs should be stated in 'related work', like $n/3\log n$ lower bound.. results, 2nlgn bit trivial encoding for normal interval and 2nlgk for k-nested (On the classes of interval grpahs of limited nesting and count of lenths)}

{\bf Succinct data structures.} There already exists a large body of work on representing various classes of graphs succinctly. This is partly motivated by theoretical curiosity and partly by the practical needs as these combinatorial structures do arise quite often in various applications. A partial list of such special graph classes would be trees~\cite{ChakrabortyS19,MunroR01}, planar graphs~\cite{AleardiDS08}, chordal graphs~\cite{MunroW18}, partial $k$-tree~\cite{FarzanK14} among others, while succinct encoding for arbitrary graphs is also considered in~\cite{FarzanM13}. Furthermore, such data structures for variety of other combinatorial objects are also well studied in the literature~\cite{BarbayHMS11,CGSS,MunroRRR12,Sadakane07}. We refer the interested reader to the recent book by Navarro~\cite{Navarro} for a comprehensive treatment of these and many more related topics on succinct/compact data structures.\\

\noindent
{\bf Algorithmic graph-theoretic results.}
For interval graphs, other than the algorithmic works mentioned earlier, there are plenty of attempts in exactly counting the number of unlabeled interval graphs~\cite{Hanlon,KP}, and the state-of-the-art result is due to~\cite{Pipp}, which is what we improve in this work. For the variants of the interval graphs that we study in this paper, there exists also a fairly large number of algorithmic results on them as well as structural results. For example, combinatorial problems like $3$-colourability~\cite{DBLP:journals/siammax/GareyJMP80}, maximum clique and independent set~\cite{DBLP:journals/jal/BhattacharyaK97, DBLP:journals/jal/GolumbicH88} can be solved in polynomial time for the circular-arc graph along with its recognition algorithm. See~\cite{Golumbic,Golumbic85} for more details regarding these combinatorial algorithms as well as various characterizations of these graph classes. 

\subsection{Our Results and Paper Organization}
\begin{table}[b]
	\centering
	\caption{Space lower/upper bounds of families of interval graphs.}
	\label{table:LB}
	\begin{tabular}{|c|c|c|c|c|}
		\hline
		Graph class & Space lower bound & Ref. & Space upper bound & Ref.\\ \hline
		interval & $n\log n- 2n\log\log n- O(n)$ & Thm.~\ref{main} & $n\log{n}+(2+\epsilon)n+o(n)$ & Thm.~\ref{interval:upper}\\
		proper/unit & $2n - O(\log n)$ & ~\cite{Finch} & $2n+o(n)$ & Thm.~\ref{proper}\\
		$k$-(im)proper  & open & & $ 2n\log {k}+6n+o(n\log{k})$ & Thm.~\ref{kproper:upper}\\
		circular arc & $n \log n - 2n \log\log n - O(n)$ & Thm.~\ref{main} & $n\log{n} + o(n\log{n})$ & Thm.~\ref{thm:circular}\\
		\hline
	\end{tabular}
\end{table}

\begin{table}[b]
	\centering
	\caption{Query times of our data structures. In what follows, we denote the length of the shortest path between two vertices $u$ and $v$, i.e., $|\spath{}(u,v)|$ by the parameter $t$, and $s$ denotes the function $\log{n}/\log\log{n}$.}
	\label{table:LB1}
	\begin{tabular}{|c|c|c|c|c|}
		\hline
		Graph class & $\degree{}(v)$/$\adjacent{}(u,v)$ &  $\neighbor{}(v)$ & $\spath{}(u,v)$ & Ref. \\ \hline
%		 & $\adjacent{}(u,v)$ &   &  &  \\ \hline
		interval & $O(1)$ &  $O(\degree{}(v))$& $O(t)$ & Thm.~\ref{interval:upper} \\
		proper/unit & $O(1)$ & $O(\degree{}(v))$& $O(t)$ & Thm.~\ref{proper} \\
		$k$-(im)proper & $O(\log \log{k})$ &  $O(\log \log{k} \cdot \degree{}(v))$ & $O(\log \log{k} \cdot t)$ & Thm.~\ref{kproper:upper} \\
		circular arc & $O(s)$ &  $O(\degree{}(v) \cdot s)$& $O(st)$ & Thm.~\ref{thm:circular}\\
		\hline
	\end{tabular}
\end{table}

\begin{comment}
\begin{table}[b]
	\centering
	\caption{Query times of our data structures.}
	\label{table:LB1}
	\begin{tabular}{|c|c|c|c|c|c|}
		\hline
		Graph class & $\degree{}(v)$ & $\adjacent{}(u,v)$ & $\neighbor{}(v)$ & $\spath{}(u,v)$ & Ref. \\ \hline
		interval & $O(1)$ & $O(1)$ & $O(\degree{}(v))$& $O(|\spath{}(u,v)|)$ & Thm.~\ref{interval:upper} \\
		proper/unit & $O(1)$ & $O(1)$ & $O(\degree{}(v))$& $O(|\spath{}(u,v)|)$ & Thm.~\ref{proper} \\
		$k$-proper & $O(\log \log{k})$ & $O(\log \log{k})$ & $O(\log \log{k} \cdot \degree{}(v))$ & $O(\log \log{k} \cdot |\spath{}(u,v)|)$ & Thm.~\ref{kproper:upper} \\
		$k$-improper & $O(\log \log{k})$ & $O(\log \log{k})$ & $O(\log \log{k} \cdot \degree{}(v))$ & $O(\log \log{k} \cdot |\spath{}(u,v)|)$ & Cor.~\ref{kimproper:upper} \\
		circular arc & $O(\log{n}/\log\log{n})$ & $O(\log{n}/\log\log{n})$& $O(\degree{}(v) \cdot \log{n}/ \log\log{n})$& $O(|\spath(u,v)|  \log{n}/ \log\log{n})$ & Thm.~\ref{thm:circular}\\
		\hline
	\end{tabular}
\end{table}
\end{comment}

We list all the preliminary data structures and graph theoretic terminologies that will be used throughout this paper, in 
%the next section.
Section~\ref{sec:prelim}.
Given an unlabeled interval graph $G$ with $n$ vertices, in Section~\ref{lower_bound} we first show that at least $n\log{n} - 2n\log\log n - O(n)$ bits are necessary to represent $G$, answering an open problem of Yang and Pippenger~\cite{Pipp}. More specifically, Yang and Pippenger~\cite{Pipp} showed a lower bound of $(n\log{n})/3 +O(n)$-bit for representing any unlabeled interval graph and asked whether this lower bound can be further improved. As circular-arc graphs are generalizations of the interval graphs, note that, a same lower bound result holds true for circular-arc graphs as well.
Next in Section~\ref{upper_bound}, we improve the trivial $(2n\ceil{ \log n})$-bit representation of $G$ (obtained by storing all the intervals correspond to the vertices in $G$ explicitly), by proposing an ($n\log{n} + O(n)$)-bit representation while  being able to support the relevant queries optimally, where the queries are defined as follows. For any two vertices $u, v \in G$,

\begin{itemize}
\item $\degree{}(v)$: returns the number of vertices that are adjacent to $v$ in $G$,
\item $\adjacent{}(u,v)$: returns true if $u$ and $v$ are adjacent in $G$, and false otherwise,
\item $\neighbor{}(v)$: returns all the vertices that are adjacent to $v$ in $G$, and
\item $\spath{}(u,v)$: returns the shortest path between $u$ and $v$ in $G$.
\end{itemize}

We show that all these queries can be supported optimally using our succinct data structure for interval graphs. More precisely, for any two vertices $v,u \in G$, we can answer $\degree{}(v)$ and $\adjacent{}(u,v)$ queries in $O(1)$ time, $\neighbor{}(v)$ queries in $O(\degree{}(v))$ time, and $\spath{}(u,v)$ queries in $O(|\spath{}(u,v)|)$ time. Note that our results for interval graphs not only improve the previous best space bound of $2n \log n$ bits (given by Klav{\'{\i}}k et al.~\cite{KlavikOS16}), but we can also support navigational queries optimally, a feature not present in~\cite{KlavikOS16}. Furthermore, 
%in addition to these basic set of operations, in the same section, 
in Section~\ref{app:graph},
we show how one can implement various fundamental graph algorithms in interval graphs, for example depth-first search ({\sf DFS}), breadth-first search ({\sf BFS}), computing a maximum independent set, determining a maximum clique, both time and space efficiently using our succinct representation for interval graphs. 

In Section~\ref{extend}, we extend our ideas to other variants of interval graphs, for example, {\it proper/unit interval graphs, k-proper and k-improper interval graphs, and circular-arc graphs}, and design succinct data structures for these graph classes as well along with supporting queries on them efficiently. For definitions of these graphs, see Section~\ref{extend}. Our succinct data structures (with efficient query support) for proper/unit, and $k$-(im)proper interval graphs improves the results of Klav{\'{\i}}k et al.~\cite{KlavikOS16} who designed encodings for these graph classes with no query support. More specifically, the asymptotic space consumption of their data structures is same as ours, yet we can support additionally the navigational queries efficiently. We summarize all of our results in Table~\ref{table:LB} and Table~\ref{table:LB1} respectively where Table~\ref{table:LB} captures the matching upper/lower bound on the space requirements of each of the graph classes and Table~\ref{table:LB1} lists all the query times of our data structures. Finally we conclude in Section~\ref{conclusion} with some remarks on possible future directions for exploring.

\section{Preliminaries}\label{sec:prelim}
We will use the following data structures in the rest of this paper.
%In this section, we introduce some data structures and terminologies which are used in the rest of the paper. 
\\\\
\noindent{\bf Rank and Select queries:}
Let $S = s_1, \dots, s_n$ be a sequence of size $n$ over an alphabet $\Sigma = \{0, 1, \dots, \sigma-1\}$.
Then for $1 \le i \le n$, and $\alpha \in \Sigma$, one can define $\rank$ and $\select$ queries as follows. 

\begin{itemize}
	\item $\rank_{\alpha}(S,i)$ = the number of occurrences of $\alpha$ in $s_1 \dots s_i$.
	\item $\select_{\alpha}(S,i)$ = the position $j$ where $s_j$ is the $i$-th $\alpha$ in $S$.
	%\item $\access(S, i)$ = returns $s_i$.
\end{itemize}
The following lemma shows that these operations can be supported efficiently using optimal space. 

\begin{lemma}[\cite{Clark:1996:EST:313852.314087, DBLP:conf/soda/GolynskiMR06}]
\label{rankselect}
Given a sequence $S = s_1, \dots, s_n$ of size $n$ over an alphabet $\Sigma = \{0, 1, \dots, \sigma-1\}$ for any $\sigma > 1$, for any $\alpha \in \Sigma$, there exists an $n \log{\sigma} + o(n\log{\sigma})$-bit data structure which answers $\rank{}_{\alpha}$ queries on $S$ and access any element of $S$ in $O(\log{(1+\log{(\sigma)})})$ time, and $\select{}_{\alpha}$ queries on $S$ in $O(1)$ time.
\end{lemma}

%Note that one can access any element of the input sequence (at a given index) in $O(1)$ (resp. $O(\log{\log{\sigma}})$) time with the $n+o(n)$ (resp. $n \log{\sigma} + o(n\log{\sigma})$)-bit data structure of Lemma~\ref{rankselect} \newline
%\\
\noindent{\bf Range Maximum Queries:}
Given a sequence $S = s_1, \dots, s_n$ of size $n$, for $1 \le i, j \le n$, the \textit{range maximum query} on range $[i, j]$ (denoted by $\rmax{}_S (i, j)$) returns the position $i \le k \le j$ such that $s_k$ is a maximum value in $s_i \dots s_j$ (if there is a tie, we return the leftmost such position).
One can define the \textit{range minimum queries} on range $[i, j]$ ($\rmin{}_S(i, j)$) analogously.
The following lemma shows that there exist data structures which can answer these queries efficiently using optimal space.

\begin{lemma}[\cite{DBLP:journals/algorithmica/BrodalDR12, DBLP:journals/siamcomp/FischerH11}]\label{rmq}
Given a sequence $S$ of size $n$ and for any $1 \le c \le n$, 
\begin{enumerate}
\item there exists a data structure of  size $O(n/c)$ bits, in addition to storing the sequence $S$, which supports $\rmax{}_S$ and $\rmin{}_S$ queries in $O(c)$ time while supporting access on $S$ in $O(1)$ time. 
\item there exists a data structure of  size $2n+o(n)$ bits (that does not store the sequence $S$) which supports $\rmax{}_S$ or $\rmin{}_S$ queries in $O(1)$ time.
\end{enumerate}
%(or $\rmax{}_S$)  
%Furthermore, if there exists no consecutive same value on $S$, one can answer both $\rmin{}_S$ and $\rmax{}_S$ queries in $O(1)$ time, using $3n+o(n)$ bits of space. 
\end{lemma}
%Note that same as the Lemma~\ref{rankselect}, one can access any position of the input sequence in $O(1)$ time with the data structure of Lemma~\ref{rmq}.
%whereas the data structure of Lemma~\ref{rmq} does not have this feature, i.e, one does not maintain the in.
\noindent{\bf Graph Terminology and Input Representation:} We will assume the knowledge of basic graph theoretic terminology as given in~\cite{Diestel} and basic graph algorithms as given in~\cite{CLRS}. Throughout this paper, $G=(V,E)$ will denote a simple undirected graph with the vertex set $V$ of cardinality $n$ and the edge set $E$ having cardinality $m$. We call $G$ an {\it interval graph} if (a) with every vertex we can associate a closed interval on the real line, and (b) two vertices share an edge if and only if the corresponding intervals are not disjoint (see Figure~\ref{figure:interval} for an example). It is well known that given an interval graph with $n$ vertices, one can assign intervals to vertices such that every end point is a distinct integer from $1$ to $2n$ using $O(n \log n)$ time~\cite{Hanlon}, and in the rest of this paper, we deal exclusively with such representations. Moreover,  for vertex $v \in V$, we refer to $I_v$ as the interval corresponding to $v$.  
\section{Counting the number of unlabeled interval graphs}\label{lower_bound}
%In this section, we show that at least $n\log{n} - 2n\log\log n - O(n)$ bits of space are necessary to represent an unlabeled interval graph with $n$ vertices, which improves the $(n\log{n})/3 +O(n)$-bit lower bound of Yang and Pippenger~\cite{Pip}. 
%Furthermore, this matches the upper bound of \cite{Pip} within $O(n)$ bits. In what follows, for a set $S$, we denote by ${S \choose k}$ the set of $k$-subsets of $S$.

This section deals with counting unlabeled interval graphs. Let $\mathcal{I}_n$ denote the number of unlabeled interval graphs  on $n$ vertices. This is the sequence with id A005975 in the On--Line Encyclopedia of Integer Sequences~\cite{OEIS}. Initial values of this sequence are given by Hanlon~\cite{Hanlon} but he did not prove an asymptotic form for enumerating the sequence. Answering a question posed by Hanlon~\cite{Hanlon}, Yang and Pippenger~\cite{Pipp} proved that the generating function
%\[
$\mathcal{I}(x)=\sum_{n\ge 1} \mathcal{I}_nx^n$
%\] 
diverges for any $x\not=0$ and they established the bounds
\beq{Pipp}
\frac{n\log n}{3}+O(n) \le \log \mathcal{I}_n \le n\log n+O(n).
\eeq

The upper bound in~\eqref{Pipp} follows from $\mathcal{I}_n \le (2n-1)!!=\prod_{j=1}^{n} (2j-1)$, where the right hand side is the number of matchings on $2n$ points on a line. For the lower bound, the authors showed
%\[
$\mathcal{I}_{3k} \ge k!/3^{3k}$
%\]
by finding an injection from $S_k$, the set of permutations of length $k$, to three-colored interval graphs of size $3k$. Furthermore, they left it open whether the leading terms of the lower and upper bounds in~\eqref{Pipp} can be matched, which is what show in affirmative by improving the lower bound. In other words, we find the asymptotic value of $\log \mathcal{I}_n$. In what follows, for a set $S$, we denote by ${S \choose k}$ the set of $k$-subsets of $S$.

%An exact number of interval graphs whose vertices correspond to half lines in $\mathbb{R}$ was earlier found by Klav\v{z}ar and Petkov\v{s}ek~\cite{KP}. For any $n\ge 1$, the number of such graphs on $n$ vertices is given as $2^{n-2}+2^{\floor{n/2}-1}$. Using a construction similar to the ones given in~\cite{KP} and \cite{Pip}, we improve the lower bound in~\eqref{Pip} so that the main terms of the lower and upper bounds match.In other words, we find the asymptotic value of $\log I_n$.

\begin{theorem}\label{main}
Let $\mathcal{I}_n$ be the number of unlabeled interval graphs with $n$ vertices. As $n\to \infty$, we have 
\beq{goal}
\log \mathcal{I}_n \ge n\log n- 2n\log\log n- O(n).
\eeq
\end{theorem}

\begin{proof}
We consider certain interval graphs on $n$ vertices with colored vertices. Let $k$ be a positive integer smaller than $n/2$ such that $k^2 \ge n -2k$, and $\eps$ a positive constant smaller than $1/2$.
For $1\le j\le k$, let $B_j$ and $R_j$ denote the intervals $[-j-\eps,-j+\eps]$ and $[j-\eps,j+\eps]$, respectively. 
These $2k$ pairwise-disjoint intervals will make up $2k$ vertices in the graphs we consider.
Now let $\cal W$ denote the set of $k^2$ closed intervals with one endpoint in $\{-k,\dots,-1\}$ and the other in $\{1,\dots,k\}$.
We color $B_1,\dots,B_k$ with blue, $R_1,\dots,R_k$ with red, and the $k^2$ intervals in $\cal W$ with white.

Together with $\cal S:=\{B_1,\dots,B_k,R_1,\dots,R_k\}$, each $\{J_1,\dots,J_{n-2k}\} \in {\cal W \choose n-2k}$ gives an $n$-vertex, three-colored interval graph.
For a given $\cal J=\{J_1,\dots,J_{n-2k}\}$, let $G_{\cal J}$ denote the colored interval graph whose vertices correspond to $n$ intervals in $\cal S \cup \cal J$, and let $\cal G$ denote the set of all $G_{\cal J}$.

Now let $G\in \cal G$. For a white vertex $w \in G$, the pair $(d_B(w),d_R(w))$, which represents the numbers of blue and red neighbors of $w$, uniquely determine the interval corresponding to $w$;
this is the interval $[-d_B(w),d_R(w)]$.
In other words, $\cal J$ can be recovered from $G_{\cal J}$ uniquely. Thus
%\[
$|\cal G| = {k^2 \choose n-2k}$.
%\]
Since there are at most $3^n$ ways to color the vertices of an interval graph with blue, red, and white, we have
\[
\mathcal{I}_n\cdot 3^n \ge |\cal G| = {k^2 \choose n-2k}\ge \left(\frac{k^2}{n-2k}\right)^{n-2k} \ge \left(\frac{k^2}{n}\right)^{n-2k}
\]
for any $k< n/2$. Setting $k=\floor{n/\log n}$ and taking the logarithms, we get
\[
\log \mathcal{I}_n \ge (n-2k)\log(k^2/n)-O(n)= n\log n-2n\log\log n -O(n). 
\]
\qed
\end{proof}

\begin{remark}
Yang and Pippenger~\cite{Pipp} also posed the question whether
%\[
$\log \mathcal{I}_n=Cn\log n+O(n)$
%\]
for some $C$ or not. According to Theorem~\ref{main}, this boils down to getting rid of the $2n\log\log n$ term in~\eqref{goal}. 
Such a result would imply that the exponential generating function
%\[
$J(x)=\sum_{n\ge 1} I_n x^n/n!$
%\]
has a finite radius of convergence. (As noted in~\cite{Pipp}, the bound $\mathcal{I}_n\le (2n-1)!!$ implies that the radius of convergence of $J(x)$ is at least $1/2$).
%Of course, a strong result would be finding $I_n$ asymptotically.
\end{remark}

\begin{remark}\label{cirlb}
After the publication of the conference version of this article~\cite{AcanCJS19}, it was pointed out to us by Cyrill Gavoille that an earlier paper from 2008 (much earlier than even Yang and Pipinger's sub-optimal lower bound result~\cite{Pipp}) by Gavoille and Paul~\cite{DBLP:journals/siamdm/GavoilleP08} already showed this lower bound for interval graphs, by obtaining a lower bound for labeled interval graphs. 
\end{remark}

\begin{remark}\label{cirlb}
A \textit{circular-arc graph} $G$ is defined as a graph whose vertices can be assigned to arcs on a circle so that two vertices are adjacent in $G$ if and only if their assigned arcs intersect. It is easy to see that every interval graph is a circular-arc graph. Hence, from Theorem~\ref{main}, we can also deduce that given a circular-arc graph $G$ with $n$ vertices, one needs at least $n\log n- 2n\log\log n- O(n)$ bits to represent $G$
\end{remark}

\section{Succinct representation of interval graphs}\label{upper_bound}

In this section, we introduce a succinct $n\log{n} + (2+\epsilon)n + o(n)$-bit representation of unlabeled interval graph $G$ on $n$ vertices with constant $\epsilon > 0$, and show that the navigational queries (\degree{}, \adjacent{}, \neighbor{}, and \spath{} queries) and some basic graph algorithms ({\sf BFS, DFS, PEO} traversals, proper coloring, computing the size of a maximum clique and maximum independent set) on $G$ can be answered/executed efficiently using our representation of $G$. 

\subsection{Succinct Representation of $G$}\label{represent}
We first label the vertices of $G$ using the integers from $1$ to $n$, as described in the following. By the assumption in Section~\ref{sec:prelim}, the vertices in $G$ can be represented by $n$ intervals $I = \{I_1 = [l_1, r_1], I_2 = [l_2, r_2], \dots, I_n = [l_n, r_n]\}$ where all the endpoints in $I$ are distinct integers in the range $[1,2n]$. Since there are $2n$ distinct endpoints for the $n$ intervals in $I$, every integer in $[1,2n]$ corresponds to a unique $l_i$ or $r_i$ for some $1 \le i \le n$. We assign the labels to the vertices in $G$ based on the sorted order of left endpoints of their corresponding intervals, i.e., for any two vertices $a, b \in G$, $a < b$ if and only if $l_a < l_b$.  

\begin{figure}[htp]
	\begin{center}
		\includegraphics[scale=0.4]{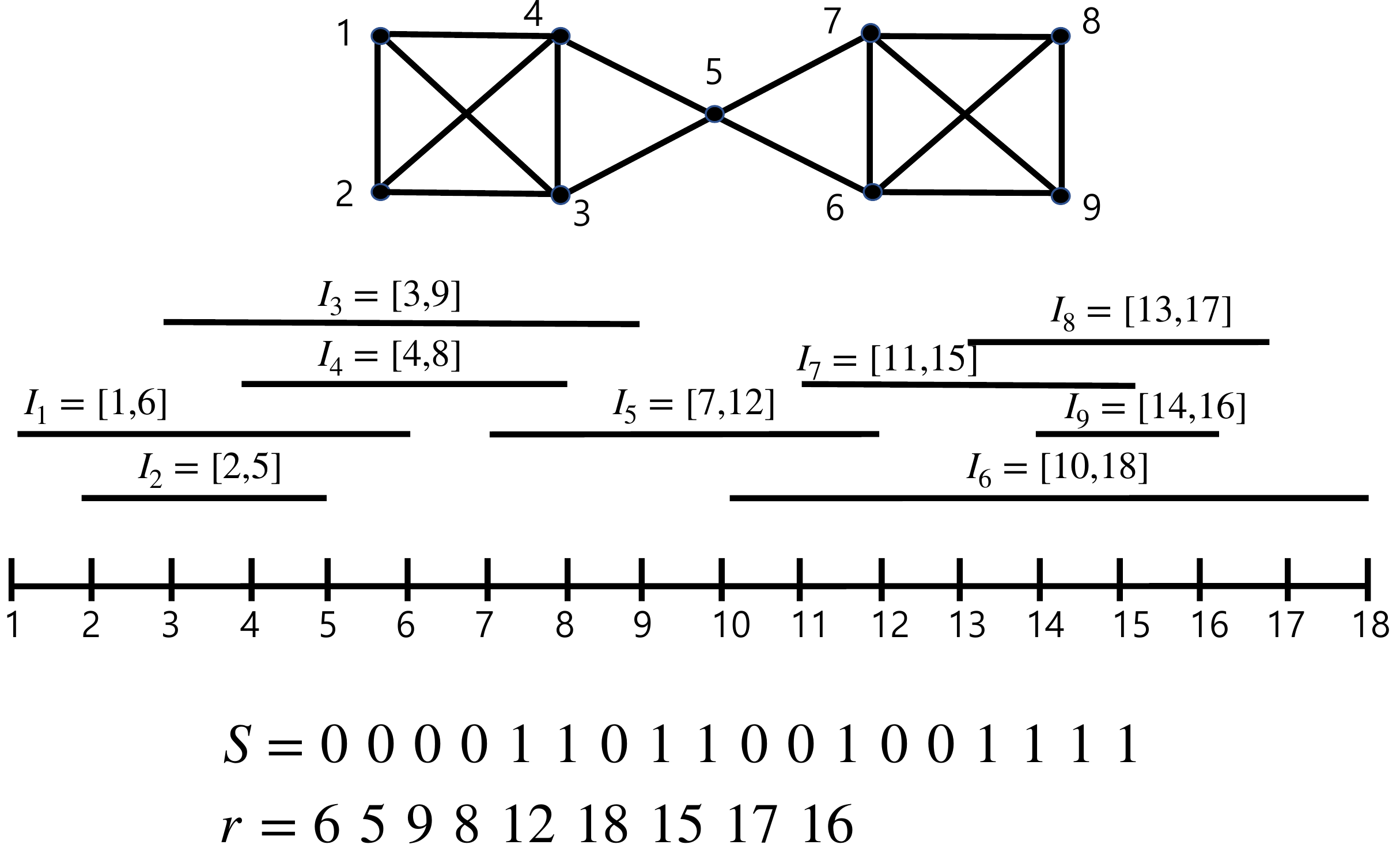}
	\end{center}
	\caption{Example of the interval graph and its representation.}
	\label{figure:interval}
\end{figure}

Now we describe the representation of $G$. Let $S = s_1 \dots s_{2n}$ be the binary sequence of length $2n$ such that
for $1 \le i \le 2n$, $s_i = 0$ if $i \in \{l_1, l_2, \dots, l_n\}$ (i.e., if $i$ corresponds to the left end point of an interval in $I$), 
and $s_i = 1$ otherwise. 
If $i = l_k$ or $i  = r_k$, we say that $s_i$ corresponds to the interval $I_k$.  
By storing the data structure of Lemma~\ref{rankselect} on $S$, we can answer $\rank{}$ and $\select$ queries on $S$ in $O(1)$ time, using $2n+o(n)$ bits.
%Also, let $len = len_1, \dots, len_n$ be a sequence of size $n$, where $len_i = r_i - l_i$, for $1 \le i \le n$. 
%We represent the sequence $len$ using the representation of~\cite{DBLP:conf/stoc/DodisPT10} that takes $\ceil{n\log{n}}$ bits 
%and supports accessing any element of $len$ in $O(1)$ time. 
%Using the representations of $S$ and $len$, it is easy to show that for any vertex $v \in G$, we can return its corresponding interval $I_v= [l_v, r_v]$ in $O(1)$ time by computing $l_v = \select_0(S, v)$, and $r_v = l_v + len_v$.
Next, we store the sequence $r = r_1 \dots r_n$, and for some fixed constant $\epsilon > 0$, we also store an $\epsilon n$-bit data structure of Lemma~\ref{rmq}(1) (with $c = 1/\epsilon$) to support $\rmax$ and $\rmin$ queries on $r$ in $O(1)$ time.
% while supporting to access $r$ in $O(1)$ time. 
Using the representations of $S$ and $r$, it is easy to show that for any vertex $v \in G$, we can return its corresponding interval $I_v= [l_v, r_v]$ in $O(1)$ time by computing $l_v = \select_0(S, v)$, and $r_v$ can be accessed from the sequence $r$.
Thus, the total space usage of our representation is $n\log{n} + (2+\epsilon)n + o(n)$ bits. See Figure~\ref{figure:interval} for an example. 

\subsection{Supporting Navigational Queries}\label{navi}
In this section, we show that $\degree{}$, $\adjacent{}$, $\neighbor{}$, and $\spath{}$ queries on $G$ can be answered in asymptotically optimal time using the representation described in the Section~\ref{represent}. 
\\\\
\noindent\textbf{$\degree{}\pmb{(v)}$ query.}
We count the number of vertices in $G$ which are not adjacent to $v$, which is a disjoint union of the two sets: (i) the set of intervals that end before the 
starting point $l_v$, and (ii) the set of intervals that start after the end point $r_v$. Using our representation the cardinalities of these two sets can be computed as follows.
The number of intervals $u$ with $r_u < l_v$ is given by $\rank_1{}(S, l_v)$. 
Similarly, the number of intervals $u$ with $r_v < l_u$ is given by $n-\rank_0{}(S, r_v)$.
Therefore, we can answer $\degree{}(v)$ query in $O(1)$ time by returning $n-\rank_1{}(S, l_v)-(n-\rank_0{}(S, r_v)) = \rank_0{}(S, r_v) - \rank_1{}(S, l_v)$.
%
%
%This is equivalent to count the all intervals which overlaps with $I_v$. 
%We count such intervals by considering considering its complement, i.e, the intervals $I_u$ such that 
%$I_u \cap I_v = \emptyset$. There are two cases for such intervals as i ) $r_u < l_v$, and ii) $r_v < l_u$. Both can be answered in $O(1)$ time by returning $rank_0(l_i-1) - rank_1(l_i-1)$ and $(rank_0(n)-rank_0(r_i+1)) - (rank_1(n) - rank_1 (r_i +1))$ respectively. 
\\\\
\noindent\textbf{$\adjacent{}\pmb{(u,v)}$ query.}
Since we can compute the intervals $I_u$ and $I_v$ in $O(1)$ time, $\adjacent{}\pmb{(u,v)}$ query can be answered in $O(1)$ by checking
$r_u < l_v$ or $r_v < l_u$ ($u$ and $v$ are not adjacent if and only if one of these conditions is satisfied).
\\\\
\noindent\textbf{$\neighbor{}(\pmb{v})$ query.}
The set of all neighbors of a vertex $v$ can be reported by considering all the intervals $I_u$ whose 
left end points are within the range $[1, \dots, r_v]$ and returning all such $u$'s with $r_u > l_v$ 
(i.e., which start to the left of $r_v$ and end after $l_v$). With our data structure, this query can be 
supported by returning the set $\{ u ~ | ~ 1 \le u \le rank_0(S, r_v) \mbox{ and } r_u > l_v \}$.
Using the \rmax{} structure stored on $r$, this can be supported in $O(\degree{}(v))$ time.
%a disjoint union of the two sets: 
%(a) $S_1 = \{ u ~ | ~ l_u < l_v < r_u \}$ (i.e., all intervals that intersect with the left end point of $I_v$), and\\
%(b) $S_2 = \{ u ~ | ~ l_v < l_u < r_v \}$ (i.e., all intervals which start within the interval $I_v$).
%
%(i) the set $S_1$ of all vertices $u$ with $I_u \cap I_v \neq \emptyset$ and $I_v \not\subset I_u$, and
%(ii) the set $S_2$ of all vertices $u$ with $I_v \subset I_u$.
%All the vertices in $S_1$ correspond to all $s_i$'s with $l_v \le i < r_v$, and moreover, every vertex whose corresponding left/right end point is in this range is adjacent the vertex $v$.
%Thus to return this set, we scan all the bits $s_i$'s in $S$ where $l_v \le i < r_v$. If $s_i =0$, then we return the vertex $\rank{}_0(S, i)$; otherwise, we return the vertex $\rank{}_0(S, i-len_i)$, if it is not already returned (this can be checked in constant time by checking whether corresponding $0$ is within the range $[l_v, r_v]$). Since $r_v-l_v \le 2 \cdot \degree{}(v)$, we can return all the intervals in this case in $O(\degree{}(v))$ time.
%% ($O(1)$ time for each $i$).
%The set $S_2$ is the set of all vertices  $u$ with $I_v \subset I_u$ : 
%It is clear that these are the vertices in $u \in \{1, \dots, v-1\}$ which satisfies $r_u > r_v$. These vertices can be returned in $O(|S_2|)$ time using the range max structure stored for the sequence $r$.
% (that stores the sequence of right end points of the intervals, in the order of their left end points).
Note that given a threshold value $l_v$ and a query range $[a,b]$ of the sequence $r$, the range max data structure can be used to report all the elements $r_u$ within the range $[a,b]$ such that $r_u > t$, in $O(1)$ time per element, using the following recursive procedure.
Compute the position $c = \rmax_r(a,b)$. If $r_c > l_v$, then return $r_c$, and recurse on the subintervals $[a,c-1]$ and $[c+1, b]$; else stop.
\\\\
\noindent\textbf{$\spath{}\mathbf{(u,v)}$ query.}
We first define the $\SUCC$ query as described in~\cite{DBLP:journals/networks/ChenLSS98}. 
For an interval $I_u$, $\SUCC(I_u)$ returns the interval $I_{u'}$ such that 
$I_u \cap I_{u'} \neq \emptyset$ and there is no 
$I_{u''}$ with $I_u \cap I_{u''} \neq \emptyset$ and $r_{u'} < r_{u''}$.
(For example in Figure~\ref{figure:interval}, $\SUCC{}(I_2)= I_3$ and $\SUCC{}(I_5)= I_6$.)
To answer the $\spath{}(u,v)$ query, let $P_{uv}$ be the shortest path from $u$ to $v$ initialized with $\emptyset$ 
(without loss of generality, we assume that $u \le v$).
If $u$ and $v$ are identical, we simply add $u$ to $P_{uv}$ and return $P_{uv}$.  
If not, we first add $u$ to $P_{uv}$ and consider two cases as follows~\cite{DBLP:journals/networks/ChenLSS98}.
\begin{itemize}
	\item If $u$ is adjacent to $v$, add $v$ to $P_{uv}$ and return $P_{uv}$.
	\item If $I_u$ is not adjacent to $I_v$, we perform	$\spath{}(\SUCC(u), v)$ query recursively.  
\end{itemize}

Since we can answer $\adjacent{}$ queries in $O(1)$ time, it is enough to show how to answer the
$\SUCC{}$ queries in $O(1)$ time. Let $k$ be the number of vertices $v$ which satisfies $l_v < r_u$, 
which can be answered in $O(1)$ time by $k = rank_0(S, r_u)$).
Then by the definition of $\SUCC{}$ query, $I_i$ with $i = \rmax_r(1, k)$ gives an answer of $\SUCC(I_u)$ if $r_i > l_u$ 
(if not, there is no vertex in $G$ adjacent to $u$).
%Since one can support range max queries on $\pi^{-1}$ sequence in $O(1)$ time with $2n+o(n)$ bits of space~\cite{DBLP:journals/siamcomp/FischerH11}, 
Therefore we can answer the $\SUCC{}$ query in $O(1)$ time, which implies $\spath{}(u, v)$ query can be answered in $O(|spath(u,v)|)$ time. Thus, we obtain a following theorem.

\begin{theorem}\label{interval:upper}
	Given an interval graph $G$ with $n$ vertices, there exists an $n\log{n}+(2+\epsilon)n+o(n)$-bit representation of $G$ which answers $\degree{}(v)$ and $\adjacent{}(u, v)$ queries in $O(1)$ time, $\neighbor{}(v)$ queries in $O(\degree{}(v))$ time, and $\spath{}(u, v)$ queries in $O(|\spath{}(u,v)|)$ time, for any vertices $u, v \in G$.
\end{theorem} 
%In the extended version, we discuss how to support some basic graph algorithms ({\sf BFS, DFS, PEO} traversals, proper coloring, computing the size of maximum clique, maximum independent set and minimum vertex cover) efficiently on $G$ with the above set of operations along with the representation of Section~\ref{represent}. 

\section{Some graph algorithms on the succinct representation of interval graphs}\label{app:graph}
The above set of operations along with the representation essentially captures the entire role of the adjacency list/array representation of the underlying interval graph. Once we have such a representation of $I$, we can talk about executing various algorithms on $I$. Here we are interested in the following set of algorithms.
\\\\
\noindent\textbf {Depth-first search ({\sf DFS}) and Breath-first search ({\sf BFS}).}
{\sf DFS} and {\sf BFS} are the two most widely known and popular graph search methods because of their versatile usage as the backbone of so many other powerful and important graph algorithms. 
In what follows, we show that essentially the vertices sorted by its ascending order of the labels i.e., $1, \dots, n$ gives both {\sf DFS} and {\sf BFS} vertex ordering of the graph $G$. Note that there may be more than one valid {\sf DFS} or {\sf BFS} ordering on $G$, but here we are interested in any of those valid and correct orderings. Moreover along the lines of recent papers~\cite{BanerjeeC0S18,ChakrabortyRS17,Chakraborty00S18,ChakrabortySR19}, here we are interested only in the ordering of the vertices in {\sf DFS} and {\sf BFS} traversals i.e., the order in which the vertices are visited for the first time during the {\sf DFS/BFS} traversal of the input graph $G$, not in actually reporting the final {\sf DFS/BFS} tree. Towards this, we show the following,

\begin{theorem} \label{dfs}
	%Given an interval graph $G$ with $n$ vertices, suppose we label the vertices of $G$ to $\{1, \dots, n\}$ to be for any vertices $a, b \in G$, $a < b$ if and only if $l_a < l_b$. Then ascending order from $1$ to $n$ gives a valid {\sf DFS} and {\sf BFS} ordering of $G$. 
	Given an interval graph $G$ with $n$ vertices, suppose we label the vertices of $G$ with $\{1, \dots, n\}$ suvh that for any $a, b \in G$, we have $a < b$ if and only if $l_a < l_b$. Then the ascending order from $1$ to $n$ gives a valid {\sf DFS} and {\sf BFS} ordering of $G$.
\end{theorem}
\begin{proof}
	We only consider the {\sf DFS} traversal in the proof (the case of {\sf BFS} traversal can be proved using a similar argument). We prove by induction on the number of visited vertices. Since we can start from an arbitrary vertex in $G$, the theorem statement holds with starting the traversal with the vertex 1. Next, suppose that we already visited the vertices $1 \dots i$ with $i < n-1$ (the case $i = n-1$ is trivial) and for every valid {\sf DFS} traversal, there exists a vertex $i' > i+1$ which is visited prior to $i+1$. This implies that there exists at least one vertex $v \in \{1, \dots, i\}$ such that $v$ is adjacent to $i'$ but not $i+1$, contradicting the fact that $l_v  < l_{i+1} < l_{i'}$. Therefore there exists a valid {\sf DFS} traversal which visits the vertex $i+1$ after visiting the vertex $i$. 
	\qed
\end{proof}

\noindent\textbf{Perfect Elimination Ordering ({\sf PEO}).}
{\sf PEO} of a graph $G$, if it exists, is defined as an ordering of the vertices of $G$ such that, for each vertex $v$, $v$ and the neighbors of $v$ that occur before $v$ in the order form a clique~\cite{Golumbic}. If we order the vertices corresponding to the intervals by sorting based on their left endpoints, then the resulting vertex order is a {\sf PEO}, as the predecessor set of every vertex forms a clique. Thus, from our representation it is trivial to generate a {\sf PEO} of the given interval graph.
\\\\
\noindent\textbf{Maximum Independent Set ({\sf MIS}) and Minimum Vertex Cover ({\sf MVC}).}
To compute an {\sf MIS}, we simulate the greedy algorithm of~\cite{DBLP:journals/networks/GuptaLL82} which works as follows. Initialize the sets $E$ and $M$ to $\emptyset$. We first find the vertex $i$ such that $r_{i}$ is the leftmost among all  the right endpoints of the intervals in $I - E$. If such an $i$ exists, we add $i$ to $M$ and add $E = E \cup I'$ where $I' \subseteq I$ is the set of all intervals whose corresponding vertices are adjacent to $i$. We repeat this procedure until no such vertex $i$ exists, and return $M$. Also {\sf MVC} can be computed from {\sf MIS} by returning the complement of {\sf MIS}, in $O(n)$ time. 
(For the graph in Figure~\ref{figure:interval}, {\sf MIS} = $\{2, 5, 9\}$ and {\sf MVC} = $\{1, 3, 4, 6, 7, 8\}$.)

Now we show how the algorithm can be implemented in time linear in the size of the input, with our representation of $G$. 
We first initialize the set $M$ to $\emptyset$ and compute $i = \rmin{}(1, n)$ (which returns the interval with the smallest right end point among all the intervals), and add vertex $i$ to $M$. Then the greedy algorithm picks the next interval with the smallest right end point in the range $[\rank_0(S, r_i)+1, n]$ of the sequence $r$.
In general, suppose $M = \{m_1, m_2 \dots m_k\}$ and $m_k$ is the last vertex added to $M$. 
Then we compute $m_{k+1} = \rmin{}(\rank_0(S, r_{m_k})+1 ,n)$, and add $m_{k+1}$ to $M$, if it exists. Thus, we can compute {\sf MIS} in time linear in the size of {\sf MIS}.\\

%scan $S$ from left to right. When we find $s_{i}$ which is the leftmost 1 in $S$, we add the vertex $m_1 = \rank{}_0(S, i-len_i)$ to $M$ in $O(1)$ time. 
%Now suppose $M = \{m_1, m_2 \dots m_k\}$ and $m_k$ is the last vertex added to $M$.
%Then we find the position $m_{k+1}$ as follows. While scanning $s_{r_{m_k}}$ to $s_n$, 
%suppose we currently access $s_i$ with $s_i = 1$.
%We then compute $m_{k+1} = \rank{}_0(S, i-len_i)$ and check whether $m_{k+1}$ is adjacent to $m_k$ or not in $O(1)$ time using $\adjacent(m_{k+1}, m_k)$ query. 
%If $m_{k+1}$ is not adjacent to $m_k$, we add $m_{k+1}$ to $M$ and repeat the procedure with the updated $M$ 
%(note that $m_{k+1}$ is also not adjacent to any vertices in $\{m_1 \dots m_{k-1}\}$ by the representation of $G$). 
%Otherwise we skip the bit $s_i$ and access the next position of $S$. After we finish the scanning of $S$, we return $M$.
%During the procedure, we scan each position in $S$ exactly once, and spend $O(1)$ time per position in $S$. 

\noindent\textbf{Computing a Maximum Clique.}
In order to find a maximum clique in $G$, we define a sequence $D = d_1, \dots, d_{2n}$ of length $2n$ where
(i) $d_1 = 1$, and (ii) for $ 1 < i \le 2n$, $d_i = d_{i-1}+1$ if $s_i = 0$ and $d_i = d_{i-1}-1$ otherwise
(for example, for the interval graph of Figure~\ref{figure:interval}, $D = 1~2~3~4~3~2~3~2~1~2~3~2~3~4~3~2~1~0$).
From the definition of $d_i$, if $s_i = 0$, there are exactly $d_i$ vertices in $G$ such that all corresponding intervals of these vertices have left endpoints at most $i$ and right endpoints larger then $i$. Thus all such $d_i$ vertices form a clique. This gives an algorithm for computing a maximum clique in $G$ as follows.
While constructing the sequence $D$ in $O(n)$ time, we maintain the index $k$ such that $d_k$ is a largest value in $D$. 
%Since it is clear that $s_k = 0$, it is enough to return all the vertices which are adjacent to the vertex $k$, which takes $O(\degree{}(k))$ time using $\neighbor{}(k)$ query. 
We then scan all the intervals and return those intervals whose left end point is at most $k$ and right end point is larger than $k$.
Therefore we can compute the maximum clique in $G$ in $O(n)$ time in total.
\\\\
\noindent\textbf{Computing a Proper Coloring.}
It is well-known that the greedy algorithm on $G$ yields the optimal proper coloring if we process the vertices of $G$ in the order of their corresponding intervals' left endpoints~\cite{Golumbic}. Thus, we simply implement this greedy coloring on $G$ from the vertex $1$ to $n$ as follows. We first maintain $n$ values $c_1, \dots, c_n$ such that for $1 \le i \le n$, $c_i \le \degree{}(i)$ stores the color of vertex $i$. Since each $c_i$ can be stored using $O(\log({\degree{}(i)}))$ bits, we can maintain all $c_i$'s using $\sum_{i = 1}^{n} O(\log({\degree{}(i)})) = O(n\log(m/n))$ bits in total by storing it as a sequence of variable length codes ($m$ denotes the number of edges in $G$). 
%\hsyn{[Couldn't understand the next sentence.]}
%The sequence $c_1, \dots, c_n$ is stored as a sequence of variable length codes, where the $i$-th code is of length $O(\log({\degree{}(i)}))$. 
To access any element of this sequence in $O(1)$ time, we store an another bit vector of size $O(n\log(m/n))$ bits, which stores a $1$ at the starting position of each code (i.e., each vertex's color), and $0$ in all other positions; and store auxiliary data structure to support $\select$ queries on it.
Now, initialize all $c_1, \dots, c_n$ to $0$ and scan the vertices from $1$ to $n$. 
While we visit the vertex $i$, we perform the $\neighbor{}(i)$ query and choose the minimum color in 
$\{1 \dots \degree{}(i)\} - \{c_v | v \in \neighbor{}(i)\}$. 
Since we use $O(\degree(i))$ time for each $\neighbor{}(i)$ query to assign the color of $i$, 
we can assign the color of all vertices in $G$ in $O(n+m)$ time, using $O(n\log(m/n))$ extra bits of space.  

Another alternative way to implement the greedy coloring on $G$ is to use a priority queue.
In this case, we first compute $\chi(G)$, which is a chromatic number of $G$. Since $G$ is an interval graph, we can compute $\chi(G)$ in $O(n)$ time on our representation by computing the size of the maximum clique of $G$. 
Now we initialize $c_1, \dots, c_n$ to $0$  and insert $1, \dots, \chi(G)$ to the priority queue $PQ$, 
and scanning $S$ from left to right.
Suppose we currently access $s_i$ which corresponds to $I_j$ (we can compute the index $j$ in $O(1)$ time).
If $s_i = 0$, we assign the minimum element of $PQ$ to $c_j$, and delete $c_j$ from $PQ$. Otherwise, we insert $c_j$ to $PQ$.
Note that we exactly perform $2n$ insert operations and $n$ delete operations on $PQ$. Therefore we can compute a proper coloring of $G$
in $O(n\log \log \chi(G))$ time using $O(n \log n)$ bits of space, using the integer priority queue structure of~\cite{DBLP:journals/jcss/Thorup04}.
Note that these two solutions use $\Omega(n)$ bits of space, With $O(n)$ bits, we cannot store the colors of all the vertices simultaneously (unless the graph is sparse), and this poses a challenge for the greedy algorithm. We leave open the problem to find a proper coloring of interval graphs using extra $O(n)$ bits.
\section{Representation of some related families of interval graphs}\label{extend}
In this section, we propose space-efficient representations for proper interval graphs, $k$-proper and $k$-improper interval graphs, and circular arc graphs. Since these graphs are restrictions or extensions (i.e., sub/super-classes) of interval graphs, we can represent them by modifying the representation in Section~\ref{represent} (to make the representation asymptotically optimal in terms of space). We also show that navigation queries on these graph classes can be answered efficiently with the modified representation. 

\subsection{Proper interval graphs}
An interval graph $G$ is \textit{proper} if there exists an interval representation of $G$ 
such that for any two vertices $u, v \in G$, $I_u \not\subset I_v$ and $I_v \not\subset I_u$
(let such interval representation of $G$ be \textit{proper representation} of $G$).
Also it is known that proper interval graphs are equivalent to the \textit{unit interval graphs},
which have an interval representation such that every interval has the same length~\cite{Roberts}. 

Now we consider how to represent a proper interval graph $G$ with $n$ vertices while supporting navigational queries efficiently on $G$. 
We first obtain an interval representation of the graph $G$ where the intervals satisfy the property of proper interval graph.
%represent $G$ properly into $n$ intervals,
We then assign labels to vertices of $G$ based on the sorted order left end points of their corresponding intervals, as described in Section~\ref{represent}.
Let $S$ be the bit sequence obtained from this representation, as defined in Section~\ref{represent}.
Then by the definition of $G$, there are no two vertices $u, v \in G$ with $l_u < l_v$ and $r_u > r_v$ (if so, $I_v \subset I_u$).
Thus by the Lemma~\ref{rankselect}, for any vertex $i \in G$ 
we can compute $l_i$ and $r_i$ in $O(1)$ time by $\select_0{}(S, i)$ and $\select_1{}(S, i)$ respectively using $2n+o(n)$ bits.
Also note that $r$ is strictly increasing sequence when $G$ is a proper interval graph, and hence one can support the $\rmax$ queries on $r = r_1 \dots r_n$ in $O(1)$ time without maintaining any data structure, by simply returning the rightmost position of the query range.
Thus, we obtain the following theorem. 

\begin{theorem}\label{proper}
Given a proper interval graph or unit interval graph $G$ with $n$ vertices, there exists a $2n+o(n)$-bit representation of $G$ which answers $\degree{}(v)$ and $\adjacent{}(u, v)$ queries in $O(1)$ time, $\neighbor{}(v)$ queries in $O(\degree{}(v))$ time, and $\spath{}(u, v)$ queries in $O(|\spath{}(u,v)|)$ time, for any vertices $u, v \in G$.
\end{theorem}   

It is known that there are asymptotically $\frac{1}{8\kappa\sqrt{\pi}}n^{-3/2}4^n$ non-isomorphic unlabeled unit interval graphs with $n$ vertices,
% when $n$ goes $\infty$, 
for some constant $\kappa > 0$~\cite{Finch}, and hence $2n -O(\log n)$ bits is an information-theoretic lower bound on representing an arbitrary proper interval graph. Thus our representation in Theorem~\ref{proper} gives a succinct representation for proper interval graphs. 

\subsection{$k$-proper and $k$-improper interval graphs}
One can generalize the proper interval graph to the following two sub-classes of interval graphs.
An interval graph $G$ with $n$ vertices is a \textit{$k$-proper interval graph} (resp. \textit{$k$-improper interval graph}) if there exists an interval representation of $G$ such that for any vertex $v \in G$, $I_v$ {\emph is contained by} (resp., {\emph contains}) at most $k \le n$ intervals in $G$ other than $I_v$. We call such an interval representation of $G$ as the \textit{$k$-proper representation} (resp. \textit{$k$-improper representation}) of $G$. Note that every proper interval graph is both a 0-proper and a 0-improper graph. The graph in Figure~\ref{figure:interval} is a 2-proper, and a 3-improper graph.

Now we consider how to represent a $k$-proper interval graph $G$ with $n$ vertices and support navigation queries efficiently on $G$.
We first represent $G$ $k$-properly into $n$ intervals, and assign the labels to vertices of $G$ based on the sorted order of their left end points, as described in Section~\ref{represent}.
Same as the representation in Section~\ref{represent}, we first maintain the data structure for supporting $\rank{}$ and $\select{}$ queries on $S$ in $O(1)$ time, using $2n+o(n)$ bits in total. 
Also we maintain the $2n+o(n)$-bit data structure of Lemma~\ref{rmq} on $r = r_1, \dots, r_n$ for supporting $\rmax$ queries on $r$ in $O(1)$ time. 
Next, to access $r$ without using $n\log{n}$ bits, we define the sequence $T = t_1 \dots t_{2n}$ of size $2n$ over the alphabet $\{0, \dots , 2k+1\}$ such that $t_i = 2k'$ (resp. $t_i = 2k'+1$) if $s_i = 0$ (resp. $s_i = 1$) 
and its corresponding interval is contained by $k' \le k$ intervals in $I=\{I_1 \dots I_n\}$. 
Now for any $0 \le i \le k$, let $R_i \subset I$ be the set of all intervals such that 
for any $[a, b] \in R_i$, $t_a = 2i$ and $t_b = 2i+1$. It is easy to show that each $R_i$ corresponds to a proper interval graph.
For example the graph in Figure~\ref{figure:interval} is a 2-proper interval graph, and
$T = 0~2~0~2~3~1~0~3~1~0~2~1~2~4~3~5~3~1$, $R_0 = \{I_1, I_3, I_5, I_6\}$, $R_1 = \{I_2, I_4, I_7, I_8\}$, and $R_2 = \{I_9\}$.  
Using $2n\log {(2k+2)}+o(n\log{k}) = 2n\log {k}+2n+o(n\log{k})$ bits, we can maintain the data structure of Lemma~\ref{rankselect} on $T$ to support $\rank{}$ queries in $O(\log{\log{k}})$ time, and $\select{}$ queries in $O(1)$ time.
%By Lemma~\ref{rankselect}, we can maintain $T$ using $2n\log {(2k+2)}+o(n\log{k}) = 2n\log {k}+2n+o(n\log{k})$ bits with supporting $\rank{}$ and $\select{}$ queries in $O(\log{\log{k}})$ and $O(1)$ time respectively.
Then for any vertex $v \in G$, we can answer its corresponding interval $I_v =[l_v , r_v]$ in $O(\log \log{k})$ time by 
$l_v = \select_0(S, v)$ and $r_v = \select_{(t_{l_v}+1)}(T, \rank_{t_{l_v}}(T, l_v))$.

Note that we can represent $k$-improper interval graphs in same space with same query time as we did for $k$-proper interval graphs by changing the definition of $T$ to be  
$t_i = 2k'$ (resp. $t_i = 2k'+1$) if $s_i = 0$ (resp. $s_i = 1$) 
and its corresponding interval contains $k' \le k$ intervals in $\{I_1 \dots I_n\}$. Thus, we obtain the following theorem.

\begin{theorem}\label{kproper:upper}
	Given a $k$-(im)proper interval graph $G$ with $n$ vertices, there exists a $(2n\log {k}+6n+o(n\log{k}))$-bit representation of $G$ which answers $\degree{}(v)$ and $\adjacent{}(u, v)$ queries in $O(\log \log{k})$ time, $\neighbor{}(v)$ queries in $O(\log \log{k} \cdot \degree{}(v))$ time, and $\spath{}(u, v)$ queries in $O(\log \log{k} \cdot |\spath{}(u,v)|)$ time, for any vertices $u, v \in G$.
\end{theorem}

\subsection{Circular-arc graphs}\label{sec:circle}
In this section, we propose a succinct representation for circular-arc graphs,
% $G$ with $n$ vertices, 
and show how to support navigation queries efficiently on the representation. Note that,
a \textit{circular-arc graph} $G$ is a graph whose vertices can be assigned to arcs on a circle so that two vertices are adjacent in $G$ if and only if their assigned arcs intersect. Now suppose that $G$ is represented by the circle $C$ together with $n$ arcs of $C$.
For an arc, we define its start point to be the unique point on it such that the arc continues from that point in the clock-wise direction 
but stops in the anti-clockwise direction; and similarly define its end point to be the unique point on it such that the arc stops in the 
clockwise direction but continues in the anti-clockwise direction.
As in the case of interval graphs, we assume, without loss of generality, that all the start and end points of all the arcs are distinct. 
We label the vertices of $G$ with the integers form $1$ to $n$ as described below. 
We first select an arbitrary arc, and label the vertex (and the arc) corresponding to this arc by $1$. 
We then traverse the circle from 
the starting point of that arc in the clockwise direction, and label the remaining vertices and arcs in the order in which 
their starting points are encountered during the traversal, and finish the traversal when we return to the starting point of the first arc.
We also map all the start and end points of all arcs, in the order in which they are encountered in the above traversal, into the range $[1, \dots, 2n]$
(since the start and end points of all the $n$ arcs are distinct). Note that all these steps can be performed in $O(n+m)$ time, where $m$ is the number of edges in $G$~\cite{DBLP:journals/algorithmica/McConnell03}. With the above defined labeling of the arcs, and the numbering
of their start and end points, let $l_i$ and $r_i$ start and end points of the arc labeled $i$, for $1 \le i \le n$.
Now the arcs can be thought of as two types of intervals in the range $[1, \dots, 2n]$;
we call an interval (and its corresponding vertex) $i$ as {\em normal} if $l_i < r_i$ (i.e., we traverse $l_i$ prior to $r_i$), and {\em reversed} otherwise.
A normal vertex $i$ corresponds to the interval $[l_i, r_i]$, while a reversed vertex $i$ actually corresponds to the union of the two intervals $[1, \dots, r_i]$ and $[l_i, \dots, 2n]$.
See Figure~\ref{figure:circular} for an example; vertex 4 and 7 are reversed, while the others are normal.
%
%For vertex $i \in G$, let $\mathcal{A}_i$
%be the arc corresponding to the vertex $i$, and let the two endpoints of $\mathcal{A}_i$ be $l_i$ and $r_i$, where
%the left end point $l_i$ is prior to right end point $r_i$ in the clockwise direction. 
%To assign the labels to the vertices of $G$,
%we first select an arbitrary arc in $C$, and label the vertex corresponding to this arc by $1$, and traverse the circle from 
%the $l_1$ in the clockwise direction. We label the remaining vertices in the order of their left end points during the traversal, and finish the traversal when we visit $l_1$ again. 
%Also for any two points $p$ and $q$ on $C$, we define an ordering on them such that $p \prec q$ if we visit $p$ before $q$ during the traversal above. 
Our representation of $G$ consists of the following substructures. 
\begin{figure}[htp]
	\begin{center}
		\includegraphics[scale=0.4]{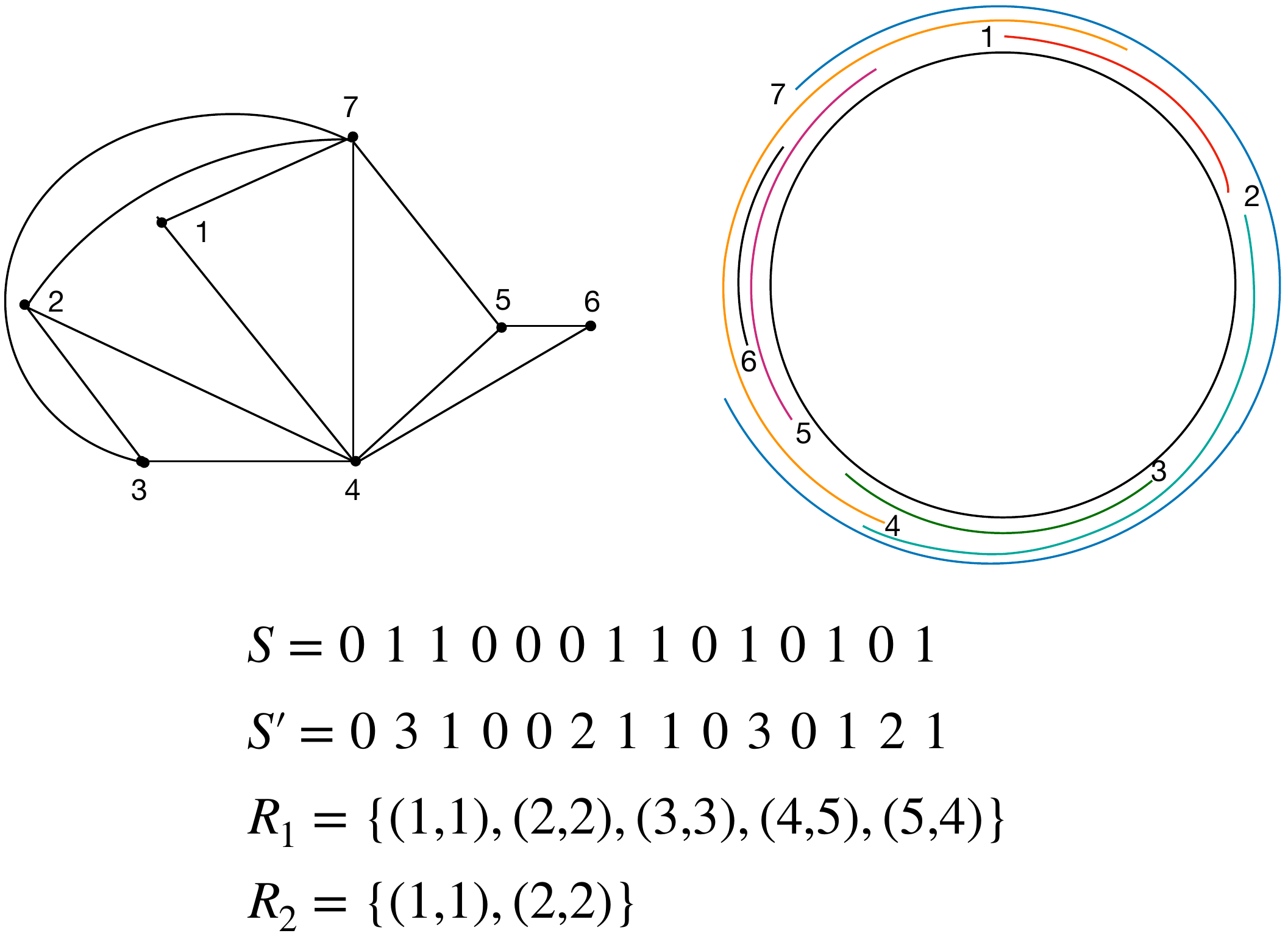}
	\end{center}
	\caption{Example of the circular graph and its representation.}
	\label{figure:circular}
\end{figure}
%In the rest of this section, we assume that all $l_1, \dots  l_n$ and $r_1 \dots r_n$ are in $\{1, \dots, 2n\}$. 
\begin{enumerate}
\item Define a binary sequence $S = s_1, \dots, s_{2n}$ of length $2n$ such that for $1 \le i \le 2n$, 
$s_i = 0$ (resp. $s_i = 1$)  if $i$-th end point encountered during the traversal of $C$ is in $\{l_1, \dots, l_n\}$ (resp. $\{r_1, \dots, r_n\}$).
%for some $1 \le k \le n $. 
%In this case, we say that the bit $s_i$ corresponds to the end point $l_k$ (resp. $r_k$), and to the arc $\mathcal{A}_k$. 
%Note that the definition of $S$ implies a mapping $f$ from end point $p \in \{l_1, \dots , l_n, r_1, \dots, r_n\}$ of $C$ to $\{1, \dots, 2n\}$ such that $f(p)$ is the position in $S$ where $s_{f(p)}$ is corresponding to $p$. 
%Let $f(l_i) = l'_i$ and $f(r_i) = r'_i$ for $1 \le i \le n$.
Now, construct a sequence $S' = s'_1, \dots, s'_{2n}$ of size $2n$ over an alphabet $\{0,1,2,3\}$ such that for all $1 \le i \le 2n$, $s'_i = s_i + 2$ if the position $s_i$ corresponds to the start or end point of a reversed interval, and $s'_i = s_i$ otherwise (i.e., if $s_i$ corresponds to a normal interval).
We represent $S'$ using the structure of Lemma~\ref{rankselect}, using $4n + o(n)$ bits, so that we can answer $\rank$
and $\select$ queries on $S'$ in $O(1)$ time. In addition, we also store auxiliary structures (of $o(n)$ bits) on top of $S'$ to support $\rank$ and $\select$ queries on $S$ (without explicitly storing $S$ -- note that one can efficiently reconstruct any subsequence of $S$ from $S'$).

\item To store the interval end points efficiently, we introduce two 2-dimensional grids of points, $R_1$ and $R_2$, defined as follows.
Suppose there are $q \le n$ normal vertices in $G$ and $n-q$ reversed vertices.
%whose corresponding left endpoint is prior to the right endpoint of the arc in $C$. 
Then let $R_1$ be a set of $q$ points on the 2-dimensional grid $[1, q] \times [1, q]$ 
which consist of $(\rank_0(S', l_i), \rank_1(S', r_i))$, for all $1 \le i \le n$  with $l_i < r_i$. 
Similarly let $R_2$ be a set of $n-q$ points on the 2-dimensional grid $[1, n-q] \times [1, n-q]$ 
which consist of $(\rank_2(S', l_i), \rank_3(S', r_i))$, for all $1 \le i \le n$ with $r_i < l_i$.
%\textbf{ADD FIGURE}
Given a set of points $R$ on 2-dimensional grid, we define the following queries (for any rectangular range $A$):
\begin{itemize}
	%\item $X(P, y)$ : returns $x$ with $(x, y) \in P$.
	\item $Y(R, x)$: returns $y$ with $(x, y) \in R$.
	\item $count(R, A)$: returns the number of points in $R$ within the range $A$. 
\end{itemize}
%By Lemma~\ref{rankselect}, we can maintain $S$, $S'$ in $6n + o(n)$ bits in total which can answer $\rank$ and $\select$ queries on $S$ and $S'$, and access any position of $S$, $S'$ in $O(1)$ time. 
We represent $R_1$ and $R_2$ using $n\log{n} + o(n\log{n})$ bits in total, such that $Y$ and $count$ queries can be supported in $O(\log{n} / \log\log{n})$ time~\cite{DBLP:conf/wads/BoseHMM09}. 
%Furthermore, since both $R1$ and $R2$ are permutations over $\{1, \dots, k\}$ and $\{1, \dots, n-k\}$ respectively, we can support $X$ queries on them in $O(\log{n})$ time using $o(n\log{n})$ additional bits of space~\cite{DBLP:journals/tcs/MunroRRR12}
%\footnote{In \cite{DBLP:journals/tcs/MunroRRR12}, one can construct an $n(\log{n} + \log{t})/t$-bit index to support $X$ queries in $O(t\log{n}/\log\log{n})$ time. Thus we can choose $t$ as any increasing function on $n$ to maintain the succinct space. In the main text, we choose $t =O(\log\log{n})$.}. 
Thus when the vertex $i$ is given, we can compute $l_i$ and $r_i$ in $O(\log{n}/\log\log{n})$ time, using the equations $l_i = \select_0{}(S, i)$, and $r_i  = \select_1{}(S',Y(R_1, \rank_0(S',l_i)))$ if $S'_{l_i} = 0$ (i.e., if $l_i$ is the left end point of a normal interval), and $r'_i  = \select_3{}(S', Y(R_2, \rank_2(S',l'_i)))$ otherwise (i.e., if $l_i$ is the left end point of a reversed interval). 
%Also when $r'_i$ is given, we can compute $l'_i$ in $O(\log{n})$ time without the index $i$ by returning $\select_0{}(S, X(R_1, \rank_1{}(S', r'_i)))$ if $S'_{r'_i} = 1$, and $\select_2{}(S, X(R_1, \rank_3{}(S', r'_i)))$ otherwise.
\item Finally, let $r' = r'_1, \dots, r'_q$ be a sequence such that for $1 \le i \le q$, 
$r'_i = r_{j_i}$ with $j_i = \select_{0}(S', i)$. 
Similarly, let $r'' = r''_1, \dots, r''_{n-q}$ be a sequence such that for $1 \le i \le n-q$, 
$r''_i = r_{j_i}$ with $j_i = \select_{2}(S', i)$.
For example, for the circular-arc graph of Figure~\ref{figure:circular}, 
$r = 3~7~8~2~14~12~10$, $r' = 3~7~8~14~12$, and $r'' = 2~10$.
Then we maintain the data structure of Lemma~\ref{rmq} on $r'$ and $r''$, using a total of $2n+o(n)$ bits, to support $\rmax$ queries on each of them.    
Thus, the overall representation takes $n\log{n} + o(n\log{n})$ bits in total.
\end{enumerate}
One can show that this representation supports $\degree$, $\adjacent$, $\neighbor$ and $\spath$ queries efficiently,
to prove the following theorem.
%Now we prove the following theorem 
%(See Appendix~\ref{app:cir} for the proof).

\begin{theorem}\label{thm:circular}
Given a circular arc graph $G$ with $n$ vertices, there exists an $(n\log{n} + o(n\log{n}))$-bit representation of $G$ which supports $\degree{}(v)$ and $\adjacent{}(u, v)$ queries in $O(\log{n}/\log\log{n})$ time, $\neighbor{}(v)$ queries in $O(\degree{}(v) \cdot \log{n}/ \log\log{n})$ time, and  $\spath{}(u, v)$ queries in $O(|\spath(u,v)|  \log{n}/ \log\log{n})$ time for any two vertices $u, v \in G$.
\end{theorem}
\begin{proof}
	Suppose we have the $n\log{n} + o(n\log{n})$-bit representation described in Section~\ref{sec:circle}.
Now we consider the following queries, and show how to support them efficiently, by extending the proof in Section~\ref{navi}.
\\\\
\noindent\textbf{$\degree{}\pmb{(v)}$ query.}
To answer $\degree{}(v)$ query, we count the number of vertices $u$ which adjacent to $v$ by considering $u$ into two cases as 
(i) $u$ is normal and
(ii) $u$ is reversed, and return the sum of them. 
Now we consider the two cases by the type of $v$ as follows.
\begin{itemize}
	\item \textbf{$\bm v$ is normal:} We can count the number of  vertices in (i) in $O(\log{n}/ \log\log{n})$ time by returning $\rank_0{}(S, r_v)-\rank_1{}(S, l_v)$, same as answering $\degree{}$ queries on interval graphs (see Section~\ref{navi}).
	Next, we count the vertices $u$ in (ii) by considering three cases as 1) $l_u < l_v$, 2) $r_v < r_u$, and 3) $l_v < l_u < r_v$ or $l_v < r_u < r_v$ separately and return the sum of them. 
	First, number of vertices in case 1) and 2) can be easily answered in $O(\log{n}/ \log\log{n})$ time by returning 
	$\rank_2{}(S', l'_v)$ and $\rank_3{}(S', r'_v)$ respectively. 
	To count the number of vertices in case 3), we first count the number of start and end points of reversed intervals between $l_v$ and $r_v$  by returning $(\rank_2{}(S', r_v)- \rank_2{}(S', l_v))+(\rank_3{}(S', r_v)- \rank_3{}(S', l_v))$. After that we subtract the number of vertices whose both start and end points exist between $l_v$ and $r_v$, which is $count(R_2, R)$ where $R = [\rank_2{}(S', l_v), \rank_2{}(S', r_v)] \times [\rank_3{}(S', l_v), \rank_3{}(S',r_v)]$. Thus we can count the number of vertices in this case in $O(\log{n}/ \log\log{n})$ time.
	%	
	%	
	%	
	%		Also we can count the number of vertices in 3) in $O(\log{n}/ \log\log{n})$ time by returning $(\rank_2{}(S', r'_v)- \rank_2{}(S', l'_v))+(\rank_3{}(S', r'_v)- \rank_3{}(S', l'_v))-(count(R_2, [\rank_2{}(S', l'_v), \rank_2{}(S', r'_v)] \times [\rank_3{}(S', l'_v), \rank_3{}(S',r'_v)]))$. Note that the first and the second term indicates the number of vertices whose corresponding left and right endpoints are in $\mathcal{A}_{v}$ respectively, and the last term indicates the number of vertices $u$ with $l'_v < r'_u < l'_u <  r'_v $ which are counted twice without it. 
	\item \textbf{$\bm v$ is reversed:} We count the number of vertices $u$ in case (i) by considering three cases as 1) $r_u < r_v$, 2) $l_v < l_u$, and 3) $r_v < l_u < l_v$ or $r_v < r_u < l_v$ separately and return the sum of them. This can be answered in $O(\log{n}/ \log\log{n})$ time by the same argument as above. For counting the vertices in (ii), we simply return $\rank_2{}(S', 2n)-1$ since all the vertices corresponds to the reverse interval cross $l_1$ in $C$, i.e., all such vertices form a clique in $G$.	
\end{itemize}
\noindent\textbf{$\adjacent{}\pmb{(u,v)}$ query.} This can be answered in $O(\log{n}/ \log\log{n})$ time by checking $l_u$, $r_u$, $l_v$, and $r_v$.
\\\\
\noindent\textbf{$\neighbor{}(\pmb{v})$ query.} 
We only describe how to answer the vertices $u$ adjacent to $v$ when $v$ is normal. 
The case when $v$ is reversed can be handled similarly. 
First we can return the all normal vertices $u$ adjacent to $v$ in $O(\log{n}/ \log\log{n} \cdot \degree(v))$ time using the same argument in the \neighbor{} query on interval graphs (see Section~\ref{navi}).
%with additional $O(\log{n}/ \log\log{n})$ time to find $(l'_v , r'_v)$. 
Next, the set of reversed vertices $u$ adjacent to $v$, is a disjoint union of the following two sets:
1) the set $S_1$ of all vertices $u$ with $l_u < r_v$, and 
2) the set $S_2$ of all vertices $u$ with $l_v < r_u$. 
We can answer all the vertices in $S_1$ in $O(\degree(v))$ time by returning $\rank_0(S, \select_2{}(S', 1)), \dots , \rank_0(S, \select_2{}(S', \rank_2{}(S', r_v)))$, 
which takes $O(1)$ time per each element.  
%Similarly, we can also answer the vertices in $S_2$ in $O(\log{n}/ \log\log{n} \cdot \degree(v))$ time.
%For answering the vertices in $S_3$, we scan $S'$ from $s_{l'_v}$ to $s_{r'_v}$ and 
%for $l'_v \le i \le r'_v$, we return the vertex $\rank_2{}(S, i)$ for every $i$ with $s_i = 2$. 
%Since $|l'_v - r'_v| \le 2 \cdot \degree(v)$, we can answer all the vertices in this case in $O(\degree(v))$ time.
Finally vertices in $S_2-S_1$ is equivalent to the the vertices $u$ in $\{\rank_0{}(S, r_v)+1, \dots, n\}$
with $l_v < r_u$. 
Using the data structure  $\rmax$ on $r''$ with a query range $[\rank_2{}(S, r_v)+1, \dots, n-q]$ on $r''$, 
these vertices  can be answered in $O(\log{n}/ \log\log{n})$ time per element by the same procedure to answer the $\neighbor{}$ queries on interval graphs.
Thus, we can answer $\neighbor(v)$ query in $O(\log{n}/ \log\log{n} \cdot \degree(v))$ time in total.   
\\\\
\noindent\textbf{$\spath{}(\pmb{u, v})$ query.}
We simulate the algorithm of~\cite{DBLP:journals/networks/ChenLSS98} with our representation of $G$. 
We first define $SUCC$ query on circular-arc graph $G$ and show how to answer the $SUCC(u)$ query in $O(\log{n}/\log\log{n})$ time.
For a set of vertices $V$ of $G$, let $V_1$ and $V_2$ be the set of normal and reversed vertices in $V$ respectively. 
Then for vertex $u \in V$, we can define $SUCC(u)$ as follows. 
\begin{itemize}
	\item If there exists a vertex in $V_2 \cap V_u$ where $V_u = \{u' | l_{u'} < r_u\}$, $SUCC(u)$ returns a vertex $u' \in V_2  \cap V_u$
	with the arc $u$ and $u'$ are intersect, and there is no vertex $u'' \in V_2  \cap V_u$ with the arc $u$ and $u''$ are intersect and $r_{u'} < r_{u''} < r_u$. Let this vertex be $u_1$. 
	\item Otherwise, $SUCC(u)$ returns a vertex $u' \in V_1$ with the arc $u$ and $u'$ are intersect, and there is no vertex $u'' \in V_1$ with the arc $u$ and $u''$ are intersect and $r_u < r_{u'} < r_{u''}$. Let this vertex be $u_2$. 
\end{itemize}
To answer $u_1$, we consider two cases as follows.
If $u \in V_1$, 
We can find $u_1$ by
returning $\rank_0{}(S, \select_2{}(S', v'))$ where $v' = \rmax_{r''}(1, \rank_2{}(S', r_u))$, which can be answered in $O(\log{n}/\log\log{n})$ time. 
Similarly if $u \in V_2$, we  can find $u_1$ in $O(\log{n}/\log\log{n})$ time by returning 
$\rank_0{}(S, \select_2{}(S', v'))$ where $v' = \rmax_{r''}(1, \rank_2{}(S', l_u))$.
%since all vertices in $V_2$ are adjacent each other.
Also we can find $u_2$ in  $O(\log{n}/\log\log{n})$ time by the same argument for answering $SUCC$ queries on interval graphs.
%After that, we find the $SUCC(u)$ on the induced sub-graph of $G$ on $V_2 - V_u$ (let this vertex be $u_2$)
%where $V_u = \{u' | r_u \prec l_{u'}\}$. If such $u_2$ exists, $SUCC(u) = u_2$ and $SUCC(u) = u_1$ otherwise.

To answer the $\spath(u,v)$ query, we do a same procedure for answering $\spath{}(u, v)$ and $\spath{}(v, u)$ queries on interval graphs (with the $SUCC$ function defined on circular arc graphs) in parallel, and return one of them which completes the procedure earlier. 
Since we can answer $SUCC$ query in $\log{n}/\log\log{n}$ time, we can answer $\spath(u,v)$ query in 
$O(\log{n}/\log\log{n} \cdot |\spath{}(u, v)|)$ time. 
% which proves the theorem.
\qed
\end{proof}
It is easy to prove that we can answer $Y$ and $count$ queries on $R_1$ and $R_2$ in $O(\log{n})$ 
time with $n\log{n}+o(n \log{n})$ bits of space by maintaining the wavelet tree~\cite{DBLP:journals/jda/Navarro14}
on  $r'$ and $r''$, instead of maintaining the data structure of~\cite{DBLP:conf/wads/BoseHMM09} on $R_1$ and $R_2$. 
This gives a simple succinct representation of $G$ while using the same space and support $\degree$ and $\adjacent$ queries in $O(\log{n})$ time, $\neighbor$ queries in $O(\log{n}  \cdot \degree(v))$ time, and  $\spath{}(u,v)$ queries in $O(|\spath(u,v)|\log{n})$ time.

Also the difference in query time on interval graphs and circular-arc graphs comes from the fact that
(i) when $\mathcal{A}_v$ is given, we need to know the number of arcs which are fully contained in $\mathcal{A}_v$ on circular-arc graph to answer $\degree(v)$ query, and (ii) since we use the data structure of~\cite{DBLP:conf/wads/BoseHMM09}, we need $O(\log n / \log \log n)$ time to access any element in $r'$ and $r''$. Thus, the query time can be improved by maintaining an $n\log{n}+O(n/c)$-bit data structure on $r'$ and $r''$ to support $\rmax$ on them, instead of maintaining the data structure of \cite{DBLP:conf/wads/BoseHMM09} on $R_1$ and $R_2$. In this case, since we can access any elements in $r'$ and $r''$ in $O(1)$ time, we can support $\adjacent$, $\neighbor$, and $\spath{}$ queries in same time as 
interval graphs, using $n\log{n}+o(n\log{n})$ bits of space in total.
In addition, for every vertex $v \in G$, by storing $\degree(v)$ explicitly using $n\ceil{\log{n}}$ bits, we can support $\degree$, $\adjacent$, $\neighbor$, and $\spath{}$ queries in same time as interval graphs, using $2n\log{n}+o(n\log{n})$ bits of space in total. This gives us the following:

\begin{corollary}
\label{cor:circular}
Given a circular arc graph $G$ with $n$ vertices, there exists an $(n\log{n} + o(n\log{n}))$-bit representation of $G$ which supports $\adjacent{}(u, v)$ queries in $O(1)$ time, $\neighbor{}(v)$ queries in $O(\degree{}(v))$ time, and  $\spath{}(u, v)$ queries in $O(|\spath(u,v)| )$ time  for any two vertices $u, v \in G$. Also, using an additional $n \ceil{\log n}$ bits, we can support $\degree{}(v)$ queries in $O(1)$ time.
\end{corollary}

\section{Conclusion and Final Remarks}\label{conclusion}
We considered the problem of succinctly encoding an unlabeled interval graph with $n$ vertices so as to support adjacency, degree, neighborhood and shortest path queries. To this end, we designed a succinct data structure that can support these queries optimally. We also showed how one can implement various combinatorial algorithms in interval graphs using our succinct data structure in both time and space efficient manner. Extending these ideas, finally, we also showed succinct/compact data structures for multiple other variants of interval graphs. One interesting open problem is to find a lower bound on space for representing $k$-proper or improper graphs, to show succinctness or improve the space of our data structure of Theorem~\ref{kproper:upper}. Also for $k$-(im)proper and circular-arc graphs, the query times of our data structures are super constant while using succinct space, hence (probably) non-optimal and we leave them as open problems whether we can design succinct data structures for supporting these queries in constant time.

\bibliographystyle{plain}
\bibliography{ref}

\end{document}